\documentclass[11pt]{article}
\usepackage[utf8]{inputenc}
\usepackage[margin=1.3in]{geometry}
\usepackage[nottoc,numbib]{tocbibind}
\usepackage{graphicx}
\usepackage{authblk}
\usepackage{complexity}
\usepackage{bbm}
\usepackage{amsmath}
\usepackage[ruled,vlined,linesnumbered,noresetcount]{algorithm2e}
\usepackage{amsfonts}
\usepackage{amsmath}
\usepackage{amssymb}

\usepackage{amsthm}
\usepackage{mathtools}
\usepackage{capt-of}
\theoremstyle{plain}
\usepackage{bbm}
\usepackage{bm}
\usepackage{xcolor}
\usepackage[shortlabels]{enumitem} 
\usepackage{comment}
\usepackage{tikz}   
\usepackage{xcolor}
\usepackage{cite}
\usetikzlibrary{decorations.pathreplacing}
\usepackage[title]{appendix}
\usepackage[framemethod=tikz]{mdframed}
\usepackage[ruled,vlined]{algorithm2e}
\usepackage{titlesec}
\usepackage{todonotes}
\usepackage[misc]{ifsym}

\newtheorem{theorem}{Theorem}
\newtheorem{lemma}{Lemma}[section]

\newtheorem{claim}[lemma]{Claim}

\newtheorem{definition}[lemma]{Definition}

\newtheorem{remark}[lemma]{Remark}

\newtheorem*{claim*}{Claim}
\newtheorem*{proposition*}{Proposition}
\newtheorem*{lemma*}{Lemma}
\newtheorem*{problem*}{Problem}

\title{Towards Non-Uniform $k$-Center with Constant Types of Radii\thanks{Supported by the Swiss National Science Foundation project 200021-184656 ``Randomness in Problem Instances and Randomized Algorithms.''}}

\author[1]{Xinrui Jia}
\author[1]{Lars Rohwedder}
\author[1]{Kshiteej Sheth\textsuperscript{(\Letter)}}
\author[1]{Ola Svensson}
\affil[1]{School of Computer and Communication Sciences, EPFL, Switzerland}
\affil[ ]{\texttt{\{xinrui.jia,lars.rohwedder,kshiteej.sheth,ola.svensson\}@epfl.ch}}

\begin{document}
\maketitle
\begin{abstract}
In the Non-Uniform $k$-Center problem we need to cover a finite metric space using $k$
balls of different radii that can be scaled uniformly. The goal is to minimize the scaling factor.  If the number of different radii is unbounded, the problem does not admit a constant-factor approximation
algorithm but
it has been conjectured that such an algorithm exists if the number of radii is constant.
Yet, this is known only for the case of two radii.
Our first contribution is a simple black box reduction which shows that if one can handle the variant of $t-1$ radii with outliers, then
one can also handle $t$ radii.
Together with an algorithm by Chakrabarty and Negahbani for two radii with outliers, this immediately implies a constant-factor approximation algorithm
for three radii; thus making further progress on the conjecture. Furthermore, using algorithms for the $k$-center with outliers problem, that is the one radii with outliers case, we also get a simple algorithm for two radii.

The algorithm by Chakrabarty and Negahbani uses a top-down approach, starting with the larger radius and then proceeding to the smaller one.
Our reduction, on the other hand, looks only at the smallest radius and eliminates it, which suggests that a bottom-up approach is promising.
In this spirit, we devise a modification of the Chakrabarty and Negahbani algorithm which runs in a bottom-up fashion, and in this way we recover their result
with the advantage of having a simpler analysis.

\end{abstract}
\newpage
\section{Introduction}\label{Sec:intro}

Clustering is a classic topic in algorithms and theoretical computer science.
The $k$-Center problem~\cite{gonzalez1985clustering ,hochbaum1985best} is a well-studied formulation of clustering, where one wants to cover points
in a metric space with balls of minimum radius around $k$ of them.
This problem has been investigated under multiple generalizations such as with outliers~\cite{charikermoses01} and multiple color classes~\cite{DBLP:conf/esa/Bandyapadhyay0P19, jia2020fair, anegg2021technique}.
From the viewpoint of $k$-Center being a location and routing problem, classic $k$-Center represents minimizing the maximum service time, assuming the speed of service is uniform at all locations.
Chakrabarty, Goyal, and Krishnaswamy (CGK)~\cite{chakrabarty2016non} introduce the \emph{Non-Uniform $k$-Center} problem to capture varying speeds at different locations.
In other words, the $k$ balls come with radii of varying sizes.
The formal definition is as follows.
\begin{definition}[The $t$-Non-Uniform $k$-Center Problem ($t$-NU$k$C)]
\label{nukcdefn}{\rm
The input is a metric space $(X, d)$ and radii $r_1 \geq r_2 \geq \cdots \geq r_t$ with $k_i$ balls of radius $r_i$.
The objective is to find $k_i$ centers $C_i \subseteq X$, $i=1,\dotsc,t$,
so that balls of radius $\alpha r_i$ around $C_i$, $i=1,\dotsc,t$, cover all of the points in $X$ and $\alpha$ is minimized.}
\end{definition}

The \textit{Robust} $t$-NU$k$C problem is a generalization of $t$-NU$k$C to incorporate the case of outliers, i.e. where one needs to cover only a certain number of points.

\begin{definition}[Robust $t$-NU$k$C]
\label{nukcoutliersdefn}
{\rm
This problem is the same as the $t$-NU$k$C problem except for an extra parameter $m$ and one needs to cover only $m$ many of the points in $X$.}
\end{definition}
It is easy to observe that Robust $(t-1)$-NU$k$C is a special case of $t$-NU$k$C with $|X|-m$ balls of radius 0. In \cite{chakrabarty2016non}, the authors gave a $(1+\sqrt{5})$-approximation for 2-NU$k$C and a 2-approximation for Robust 1-NU$k$C. They also showed that no constant-factor approximation is possible when $t$ is part of the input, assuming $\P \neq \NP$. Further, the authors conjectured an $O(1)$-approximation to be possible when $t=O(1)$. Recently, Chakrabarty and Negahbani (CN) \cite{DBLP:conf/ipco/ChakrabartyN21} obtained an important result making progress towards this conjecture. They obtained a 10-approximation for Robust 2-NU$k$C, which is a special case of 3-NU$k$C. However, their techniques do not seem to extend for 3-NU$k$C and they state in their paper that new ideas would be needed to make further progress. We show a simple reduction in this paper from $t$-NU$k$C to Robust $(t-1)$-NU$k$C for all $t$ that loses only a constant factor in the approximation guarantee. This together with the algorithm of \cite{DBLP:conf/ipco/ChakrabartyN21} for Robust 2-NU$k$C implies a simple constant-approximation for 3-NU$k$C.
\begin{theorem}\label{tradiimain}
If there is an $\alpha$-approximation for Robust $(t-1)$-NU$k$C, then there is a $(2\alpha+2)$-approximation for $t$-NU$k$C.
\end{theorem}
Thus, the 10-approximate algorithm of \cite{DBLP:conf/ipco/ChakrabartyN21} for Robust 2-NU$k$C implies a 22-approximate algorithm for 3-NU$k$C. Since no constant approximation was known when $t\geq 3$, this makes further progress on the conjecture of \cite{chakrabarty2016non}. We also note that using the 2-approximation algorithm of \cite{chakrabarty2016non} or \cite{DBLP:conf/esa/Bandyapadhyay0P19} for $k$-center with outliers, Theorem~\ref{tradiimain} also gives a simpler alternate 6-approximation algorithm for 2-NU$k$C as compared to the algorithm of \cite{chakrabarty2016non}.

\paragraph*{Comparison of previous work and our approach.} We briefly describe the approach of \cite{DBLP:conf/ipco/ChakrabartyN21} and compare it with ours. The 10-approximation algorithm of \cite{DBLP:conf/ipco/ChakrabartyN21} for Robust 2-NU$k$C uses a connection to the \emph{firefighter on trees} problem initially developed in \cite{chakrabarty2016non} and employs a multi-layered \emph{round-or-cut} procedure using the ellipsoid algorithm. Given a fractional solution $x$ to an instance of Robust 2-NU$k$C, \cite{DBLP:conf/ipco/ChakrabartyN21} obtains an instance of a 2-layered firefighter problem with a corresponding fractional solution $y$. This firefighter instance has the property that an integral solution to it would imply an approximate solution to the initial Robust 2-NU$k$C instance. The top layer in the firefighter instance corresponds to potential centers for balls of the larger radius and the bottom layer corresponds to potential centers for balls of the smaller radius. They show that if $y$ does not put too much mass on vertices of the top layer then one can easily obtain an integral solution to the firefighter instance, which in turn gives an approximate solution for the original Robust 2-NU$k$C instance. In the other case, they show that if there is an integral solution that puts a lot of mass on the top layer, then one can reduce the original instance to a \emph{well-separated} instance with respect to to balls of the larger radii, that is, balls of the larger radii are only allowed to be placed on a specified set of points that have large pairwise distance. If this instance is infeasible, then it implies that every integral solution of the original instance does not put too much mass on the top layer of the firefighter instance. This can be used to obtain a linear inequality violated by $x$ but satisfied by every integral solution of the original instance, which is then fed back to the ellipsoid algorithm in the round-or-cut framework. \cite{DBLP:conf/ipco/ChakrabartyN21} then designs an algorithm that either returns an approximate solution for a well-separated instance with respect to the larger radius or proves that the instance is infeasible, by exploiting the fact that balls of the larger radius do not intersect and interact as they are only allowed to be centered on points that have large pairwise distance. This algorithm is again based on the round-or-cut framework.

The algorithm of \cite{DBLP:conf/ipco/ChakrabartyN21} proceeds in a top-down fashion in the following sense: They first reduce a general instance to a well-separated one with respect to the larger radius and then proceed by solving such an instance. On the contrary, our reduction that gives Theorem~\ref{tradiimain} is bottom-up. Given an instance of $t$-NU$k$C, we greedily partition the metric space into clusters of radius two times the smallest radius. These clusters are disjoint and thus the centers of these clusters are well-separated with respect to the smaller radius. This allows us to throw away information about the smallest radius and all points except the centers of these clusters to obtain an instance of Robust $(t-1)$-NU$k$C, i.e., one type of radius is eliminated. In Section~3, we show how we can also obtain a 10-approximation for Robust 2-NU$k$C that works in a bottom-up fashion as compared to the top-down approach of \cite{DBLP:conf/ipco/ChakrabartyN21}. We also use a multi-layered round-or-cut approach. In our outer layer of the round-or-cut framework using the ellipsoid algorithm we reduce a general instance to an instance where balls of the smaller radius do not interact. Then using another layer of round-or-cut we reduce such a structured instance to a well-separated instance with respect to the larger radius. We observe that such an extremely structured instance is a \emph{laminar instance} that can be solved using standard dynamic programming techniques. Our approach can be viewed as a bottom-up implementation of the algorithm of \cite{DBLP:conf/ipco/ChakrabartyN21} with a simpler analysis. In particular, we do not need to prove Lemma 4 in \cite{DBLP:conf/ipco/ChakrabartyN21}, which essentially argues that if the firefighter instance obtained from a well-separated Robust 2-NU$k$C instance with respect to the larger radius does not have an integral solution, then a certain linear inequality serves as a separating inequality. 

Using the bottom-up view rather than a top-down one, we are able to obtain a simple reduction from $t$-NU$k$C to Robust $(t-1)$-NU$k$C which in turn implies a simple constant approximation for 3-NU$k$C, thus making progress on the conjecture of \cite{chakrabarty2016non}. Secondly, with this view we are also able to design a bottom-up implementation of the 10-approximation algorithm of \cite{DBLP:conf/ipco/ChakrabartyN21} for Robust 2-NU$k$C that has a simpler analysis.

\paragraph*{Preliminaries and Notation.} In our problem we are given a metric space $(X,d)$ where $X$ is a finite set of points and $d:X\times X \rightarrow \mathbb{R}_{+}$ is a distance function that satisfies the triangle inequality. For any $v\in X$, $U\subseteq X$ and $r\geq 0$ we let $\mathcal{B}_{U}(v,r)=\{u\in U : d(u,v)\leq r\}$. By $\mathcal{B}(v,r)$ we mean $\mathcal{B}_{X}(v,r)$ and we refer to it as the ball of radius $r$ centered at $v$. For any vector $x\in \mathbb{R}^{|X|}$ and set $S\subseteq X$ we write $x(S) = \sum_{v\in S}x_{v}$. We will work with the approximate feasibility versions of the problems defined in Definitions \ref{nukcdefn} and \ref{nukcoutliersdefn}. Our algorithms for these problems will either output that the input instance is infeasible, that is there is no solution with $\alpha=1$, or output a feasible solution with some $\alpha \leq \alpha^*$. Using binary search, such an algorithm would imply an $\alpha^*$-approximation for the $t$-NU$k$C and Robust $(t-1)$-NU$k$C. Thus in this paper, by a \emph{feasible instance} of $t$-NU$k$C and Robust $(t-1)$-NU$k$C we mean an instance that has a feasible solution with $\alpha=1$.

\section{Reducing $t$-NU$k$C to Robust $(t-1)$-NU$k$C}\label{Sec:main-thm}
In this section we present our simple reduction of $t$-NU$k$C to Robust $(t-1)$-NU$k$C and its analysis, which will imply Theorem \ref{tradiimain}.
 Using Theorem \ref{tradiimain} and the recent 10-approximation algorithm for Robust 2-NU$k$C obtained in~\cite{DBLP:conf/ipco/ChakrabartyN21}, we obtain the following corollary.
\begin{theorem}
There is a 22-approximation algorithm for 3-NUkC.
\end{theorem}

We first present the algorithm for performing the reduction and then follow it with the main statement of the reduction. Then we show how to prove Theorem~\ref{tradiimain} using the reduction and we conclude the section by proving correctness of the reduction.
\vspace{0.4cm}
\\
\vspace{0.4cm}
\begin{algorithm}[H]
\SetAlgoLined
 $\textbf{Input:}$ $\mathcal{I} = ((X,d), (k_1,r_1),\ldots,(k_t,r_t))$ ,$r_1\geq \ldots \geq r_t$\;
 $\textbf{Init:}$ Set $L = \emptyset$, $U = X$ \;
 \While{$U\neq \emptyset$}{
 Let $v$ be an arbitrary point in $U$\;
 $L\leftarrow L\cup \{v\}$\;
 $\text{Child}(v) := \mathcal{B}_{U}(v,2r_t)$\;
 $U\leftarrow U\setminus \mathcal{B}_{U}(v,2r_t)$\;
 }
 
 $\textbf{Return:}$ $(L,\{\text{Child}(v)\}_{v\in L})$
 \caption{RadiiCompression}\label{RadiiCompression}
\end{algorithm}
The crucial property of the reduction is summarized in the following lemma.
\begin{lemma}\label{reduction}
Given a feasible instance $\mathcal{I} = ((X,d), (k_1,r_1),\ldots,(k_t,r_t))$ of $t$-NU$k$C, \\
RadiiCompression($\mathcal{I}$) returns $(L,\{\text{Child}(v)\}_{v\in L})$ where $L\subseteq X$ and $\{\text{Child}(v)\}_{v\in L}$ partitions $X$ such that
\begin{itemize}
    \item if $|L|\leq k_t$, then $k_t$ balls of radius $2r_t$ around points in $L$ cover $X$,
    \item otherwise, $\mathcal{I}_{reduced} = ((L,d), (k_1,2r_1),\ldots,(k_{t-1},2r_{t-1}), m=|L|-k_t)$ is a feasible instance of Robust $(t-1)$-NU$k$C.
\end{itemize}
\end{lemma}
Before we prove this lemma, let us see how it implies Theorem~\ref{tradiimain}.
\begin{proof}[Proof of Theorem \ref{tradiimain}]
We can solve the original feasible $t$-NU$k$C instance $\mathcal{I}$ as follows. We first run RadiiCompression$(\mathcal{I})$ to obtain $(L,\{\text{Child}(v)\}_{v\in L})$. Then by applying Lemma~\ref{reduction} we either obtain a $2$-approximation to $\mathcal{I}$ (if $|L|\leq k_t$) or a feasible instance $\mathcal{I}_{reduced}$ of Robust $(t-1)$-NU$k$C which then is solved using the $\alpha$-approximation algorithm assumed to exist. The algorithm returns an $\alpha$-approximate solution, i.e. sets $C_1,\ldots,C_{t-1}\subseteq L$ with $|C_i|\leq k_i$ such that $k_i$ balls of radius $2\alpha r_i$ around points in $C_i$ for all $1\leq i\leq t-1$ cover at least $|L|-k_t$ points of $L$. We increase the radius of each of these balls from $2\alpha r_i$ to $2\alpha r_i + 2r_t$. We also open at most $k_t$ balls of radius $2r_t$ around the points in $L$ not covered. Since each point in $X$ is at distance at most $2r_t$ from some point in $L$ and every point in $L$ is either an open center or is at distance at most $2 \alpha r_i$ from an open center in $C_i$ for some $1\leq i \leq t-1$, by the triangle inequality all points in $X$ are covered. We used at most $k_i$ balls of radius $2\alpha r_i+2r_t\leq (2\alpha+2)r_i$ for each $1\leq i \leq t-1$ and at most $k_t$ balls of radius $2r_t$. Thus, we get a $(2\alpha +2)$-approximation algorithm for $t$-NU$k$C, assuming there is an $\alpha$-approximation algorithm for Robust $(t-1)$-NU$k$C.
\end{proof}
Now we present the proof of Lemma \ref{reduction}.
\begin{proof}[Proof of Lemma \ref{reduction}]
Clearly $\{\text{Child}(v)\}_{v\in L}$ partitions $X$, as otherwise the while loop would not have terminated.
If $|L|\leq k_t$, then since $\{\text{Child}(v)\}_{v\in L}$ partitions $X$ and every point in $\text{Child}(v)$ is at distance at most $2r_t$ from $v$ for all $ v\in L$, we can open $|L|\leq k_t$ balls of radius $2r_t$ around points in $L$ and cover all points in $X$.\\
For the rest of the proof we assume that $|L|>k_t$. Consider a feasible solution $C_1,\ldots,C_{t}\subseteq X$ where $|C_i|\leq k_i$ for all $1\leq i \leq t$ of $\mathcal{I}$, i.e. $k_i$ balls of radius $r_i$ around points in $C_i$ for all $1\leq i \leq t$ cover $X$. Let $L_i\subseteq L$ be the points of $L$ that are covered by $C_i$. If a point is covered by $C_i$ and $C_j$ where $i<j$ then we only include it in $L_i$. Since each point in $L$ is covered, $\{L_i\}_{i=1}^t$ partitions $L$. Note that each ball of radius $r_t$ can cover at most one point in $L_t$ as the pairwise distance between any two points in $L_t\subseteq L$ is strictly more than $2r_t$. Hence, $|L_t|\leq |C_t|\leq k_t$.

By ``slightly" moving the centers $C_1,\dotsc,C_{t-1}$ we will exhibit a solution that covers all points in $L \setminus L_{t}$.
To this end, consider some ball of radius $r_i$ centered at a point $p \in C_i$ that covers a point $v\in L_i$. Then for all $u\in L_i$ also covered by the ball around $p$
we have $d(u,v)\leq d(u,p)+d(p,v)\leq 2r_i$. Thus if we move this ball to be centered at $v$ instead of $p$ and increase the radius to $2r_i$, it will cover all the points in $L_i$ that were covered by it previously when it was centered at $p$.
Repeating this procedure for every $p\in C_i$ and every $1 \le i \le t-1$,
we obtain new centers $C'_1 \subseteq L_1,\dotsc,C'_{t-1} \subseteq L_{t-1}$.
It follows that balls of radius $2r_i$ around the centers in $C'_i$, $1 \le i \le t-1$, cover at least $\sum_{i=1}^{t-1}|L_i|=|L|-|L_t|\geq |L|-k_t$ points of $L$. This exhibits feasibility of the Robust $(t-1)$-NU$k$C instance $\mathcal{I}_{reduced}$.
\end{proof}
\section{A bottom-up algorithm for Robust 2-NU$k$C}
\label{Sec:alt}
The main result in~\cite{DBLP:conf/ipco/ChakrabartyN21} is a 10-approximation for Robust 2-NU$k$C. In this section we present an alternative, bottom-up implementation of the algorithm of~\cite{DBLP:conf/ipco/ChakrabartyN21} as briefly discussed in Section~1. The main theorem we will show in this section is the following.
\begin{theorem}\label{bottomupnukc}
There is a 10-approximation for Robust 2-NU$k$C.
\end{theorem}
\paragraph*{Linear programming relaxation.} The input consists of an instance $\mathcal{I}= ((X,d),$\\$(k_1,r_1),(k_2,r_2),m)$ of Robust 2-NU$k$C with $r_1\geq r_2$. We will be working with the following natural LP formulation for the problem that we refer to as LP1.
\begin{align*}
     \text{cov}_1(v) \leq &\sum_{u \in \mathcal{B}(v,r_1)} x_{u,1} \quad \forall v \in X \\
     \text{cov}_2(v) \leq &\sum_{u \in \mathcal{B}(v,r_2)} x_{u,2} \quad \forall v \in X \\
     \text{cov}(v) := & \text{cov}_1(v)+\text{cov}_2(v)\leq 1  \quad \forall v \in X \\
     \sum_{v \in X} x_{v,1} &\leq k_1, \quad
     \sum_{v \in X} x_{v,2} \leq k_2 \\
     &\sum_{v \in X} \text{cov}(v) \geq m  \\
     x_{v,i} & \geq 0, \quad \forall v \in X, i\in\{1,2\}.\\
     \text{cov}_i(v) & \geq 0, \quad \forall v \in X, i\in\{1,2\}.
\end{align*}
For every $v\in X$, $\text{cov}_i(v)$ denotes the (fractional) amount that $v$ is covered by balls of radius $r_i$ and $x_{v,i}$ denotes the (fractional) amount that a ball of radius $r_i$ centered at $v$ is open, for $i\in \{1,2\}$.
For the instance $\mathcal{I}$, we denote by $\mathcal{P}_{\mathcal{I}}$ the convex hull of all possible integral coverages $\{(\text{cov}_1(v),\text{cov}_2(v))\}_{v\in X}$ induced by integral feasible solutions of $\mathcal{I}$.

We now proceed to the proof of Theorem \ref{bottomupnukc}. We will use the round-or-cut method on $\mathcal{P}_{\mathcal{I}}$ via the ellipsoid algorithm to solve this problem. This method was first used in \cite{carr} in the context of the minimum knapsack problem and later had been used as a successful technique in designing approximation algorithms for other problems such as clustering~\cite{an2017lp, li2017uniform, li2016approximating, chakrabarty2019generalized,anegg2021technique} and network design~\cite{chakrabarty2015approximability}. We now explain the round-or-cut method in our context. In this iterative method, we are given fractional coverages $\{(\text{cov}_1(v),\text{cov}_2(v))\}_{v\in X}$ in each iteration. Using these coverages we will either generate a 10-approximate solution to $\mathcal{I}$, or find a linear inequality violated by these coverages but satisfied by each point in $\mathcal{P}_{\mathcal{I}}$. This inequality is then fed back to the ellipsoid algorithm, which computes a new fractional solution to be
used in the next iteration. The separating hyperplanes we output will have encoding length bounded by a polynomial in $|X|$, so this procedure will terminate in a polynomial number of iterations and output an approximate solution along the way or prove that the instance is infeasible.

Recall that $\text{cov}(v) = \text{cov}_1(v)+\text{cov}_2(v)$ for all $v\in X$. First we check if $\sum_{v\in X}\text{cov}(v)\geq m$ as otherwise this inequality itself acts as a separating inequality. Now given these fractional coverages, we run the classic Hochbaum and Shmoys (HS) subroutine~\cite{hochbaum1985best} on $\mathcal{I}$. This subroutine greedily partitions $X$ into clusters of radius $r$ specified in the input. This is done by picking the point with the highest fractional coverage, removing a ball of radius $r$ around it, and repeating. A pseudocode description of this can be found in Algorithm~\ref{HScluster}. We set
\begin{equation*}
     (L_2,\{\text{Child}_2(v)\}_{v\in L_2})= \mathrm{HS}((X,d),\{\text{cov}(v)\}_{v\in X},2r_2),
\end{equation*}
where $\{\text{Child}(v)\}_{v\in L_2}$ partitions $X$ and $\text{Child}(v)$ has radius $2r_2$ for all $v\in L_2$. Let $w(v) = |\text{Child}(v)|$. The HS subroutine guarantees that $\text{cov}(v)\geq \text{cov}(u)$ for all $v\in L_2, u\in \text{Child}(v)$. Hence the coverages satisfy
\begin{equation*}
    \sum_{v\in L_2} w(v) \mathrm{cov}(v) = \sum_{v\in L_2} \sum_{u\in \mathrm{Child}(v)} \mathrm{cov}(v)\geq \sum_{v\in L_2} \sum_{u\in \mathrm{Child}(v)} \mathrm{cov}(u) = \sum_{v\in X} \mathrm{cov}(v)\geq m .
\end{equation*}
\begin{algorithm}
\SetAlgoLined
 $\textbf{Input:}$ $(X,d)$, $\{\mathrm{cov}(v)\}_{v\in X}$, $r$\;
 Let $L = \emptyset$, $U = X$ \;
 \While{$U\neq \emptyset$}{
 Let $v = \underset{u\in U}{\mathrm{argmax}}\ \mathrm{cov}(u)$\;
 $L\leftarrow L \cup \{v\}$\;
 Child$(v) =  \mathcal{B}_{U}(v,r)$\;
 $U\leftarrow U\setminus \mathcal{B}_{U}(v,r)$;
 }
 Return $(L,\{\mathrm{Child}(v)\}_{v\in L})$
 \caption{HS}\label{HScluster}
\end{algorithm}
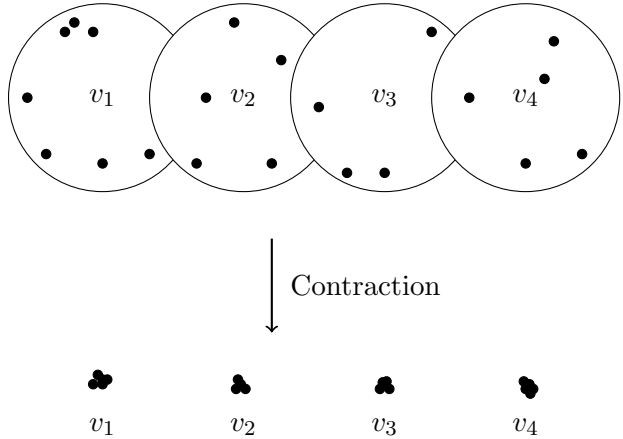
\begin{figure}
\centering

\begin{tikzpicture}[scale = 1.25]
\draw (0 ,0) circle [radius=1];
\draw [fill=white] (1.5 ,0) circle [radius=1];
\draw [fill=white] (3 ,0) circle [radius=1];
\draw [fill=white] (4.5 ,0) circle [radius=1];
\draw [fill] (-0.8, 0) circle [radius = 0.05];
\draw [fill] (-0.4, 0.7) circle [radius = 0.05];
\draw [fill] (-0.3, 0.8) circle [radius = 0.05];
\draw [fill] (-0.1, 0.7) circle [radius = 0.05];
\draw [fill] (-0.6, -0.6) circle [radius = 0.05];
\draw [fill] (0, -0.7) circle [radius = 0.05];
\draw [fill] (0.5, -0.6) circle [radius = 0.05];
\node at (0,0) {$v_1$};
\draw [fill] (1.1, 0) circle [radius = 0.05];
\draw [fill] (1.4, 0.8) circle [radius = 0.05];
\draw [fill] (1.9, 0.4) circle [radius = 0.05];
\draw [fill] (1, -0.7) circle [radius = 0.05];
\draw [fill] (1.8, -0.7) circle [radius = 0.05];
\node at (1.5, 0) {$v_2$};
\draw [fill] (2.3, -0.1) circle [radius = 0.05];
\draw [fill] (3.5, 0.7) circle [radius = 0.05];
\draw [fill] (3, -0.8) circle [radius = 0.05];
\draw [fill] (2.6, -0.8) circle [radius = 0.05];
\node at (3,0) {$v_3$};
\draw [fill] (3.9, 0) circle [radius = 0.05];
\draw [fill] (4.7, 0.2) circle [radius = 0.05];
\draw [fill] (4.8, 0.6) circle [radius = 0.05];
\draw [fill] (4.5, -0.7) circle [radius = 0.05];
\draw [fill] (5.1, -0.6) circle [radius = 0.05];
\node at (4.5,0) {$v_4$};
\draw [thick, ->] (1.8, -1.5) -- (1.8, -2.5);
\node at (2.8, -2) {Contraction};
\draw [fill] (0, -3) circle [radius = 0.05];
\draw [fill] (0, -3.05) circle [radius = 0.05];
\draw [fill] (-0.1, -3.05) circle [radius = 0.05];
\draw [fill] (-0.05, -2.95) circle [radius = 0.05];
\draw [fill] (0.05, -3) circle [radius = 0.05];
\node at (0, -3.5) {$v_1$};
\draw [fill] (1.44, -3) circle [radius = 0.05];
\draw [fill] (1.42, -3.1) circle [radius = 0.05];
\draw [fill] (1.52, -3.1) circle [radius = 0.05];
\draw [fill] (1.47, -3.05) circle [radius = 0.05];
\node at (1.5, -3.5) {$v_2$};
\draw [fill] (2.98, -3.03) circle [radius = 0.05];
\draw [fill] (3.02, -3.02) circle [radius = 0.05];
\draw [fill] (2.95, -3.1) circle [radius = 0.05];
\draw [fill] (3.05, -3.1) circle [radius = 0.05];
\node at (3, -3.5) {$v_3$};
\draw [fill] (4.48, -3.02) circle [radius = 0.05];
\draw [fill] (4.5, -3.1) circle [radius = 0.05];
\draw [fill] (4.55, -3.15) circle [radius = 0.05];
\draw [fill] (4.58, -3.1) circle [radius = 0.05];
\draw [fill] (4.54, -3.05) circle [radius = 0.05];
\node at (4.5, -3.5) {$v_4$};
\end{tikzpicture}
\caption{Example of a contraction procedure to get $\mathcal{I}_{\mathrm{contracted}}$. $v_1,v_2,v_3$, and $v_4$ are points of $L_2$, and points inside the circle centered at $v_i$ make up Child$(v_i)$.}
\label{fig:contractionfig}
\end{figure}

We will now show that if there is an integral solution satisfying $\sum_{v\in L_2}w(v)\text{cov}(v)\geq m$, we can reduce the problem to a more structured instance
\begin{equation*}
    \mathcal{I}_{\mathrm{contracted}} = ((X',d'),(k_1,2r_1),(k_2,0),m) .
\end{equation*}
The metric
$(X',d')$ is obtained by contracting each Child$(v)$, i.e. by co-locating each point in Child$(v)$ with $v$, for all $ v\in L_2$. Thus, the number of points co-located with each $v\in L_2$ is $w(v)=|\text{Child}(v)|$. We refer the reader to Figure~\ref{fig:contractionfig} for an illustration of this contraction.
\begin{remark}
Note that we will use the convention that a ball of radius 0 around any point $v$ covers all points co-located with $v$ while solving $\mathcal{I}_{\mathrm{contracted}}$.
\end{remark}
We now state the formal lemmas about the feasibility of $\mathcal{I}_{\mathrm{contracted}}$ and about approximately solving $\mathcal{I}_{\mathrm{contracted}}$ or determining its infeasibility.
\begin{lemma}\label{contractionreduction}
Either the  Robust 2-NU$k$C instance $\mathcal{I}_{\mathrm{contracted}}$ is feasible, or the inequality\\ $\sum_{v\in L_2}w(v)\mathrm{cov}(v)<m$ separates $\{(\mathrm{cov}_1(v),\mathrm{cov}_2(v))\}_{v\in X}$ from $\mathcal{P}_{\mathcal{I}}$.
\end{lemma}
We summarize the key properties of $I_{\mathrm{contracted}}$ in the following definition.
\begin{definition}{\rm
An instance $\mathcal{I}=((X,d),(k_1,r_1),(k_2,r_2),m)$ of Robust 2-NU$k$C is called a \emph{contracted} instance if we are given an additional set $L\subseteq X$ in the input and $\mathcal{I}$ and $L$ satisfy the following properties.
\begin{enumerate}
    \item $r_2=0$.
    \item For every point $u\in X$ there is a point $v\in L$ such that $d(u,v)=0$. Furthermore, $d(v,v')>0$ for every $v,v'\in L$.
\end{enumerate}}
\end{definition}
\begin{lemma}\label{algoforcontracted}
There is a polynomial time 4-approximation algorithm for contracted instances. 
\end{lemma}
Before we prove these lemmas, let us see how they imply Theorem~\ref{bottomupnukc}. In the current iteration of round-or-cut we have constructed $\mathcal{I}_{\mathrm{contracted}}$ using the fractional coverages as described above. Then we apply Lemma~\ref{algoforcontracted} to $\mathcal{I}_{\mathrm{contracted}}$ as it is a contracted instance with $L=L_2$: If we obtain a 4-approximate solution $C'$ for $\mathcal{I}_{\mathrm{contracted}}$ then we can increase the radius of each ball in $C'$ by $2r_2$ and get a solution for $\mathcal{I}$ that covers at least $m$ points. This is because the radius of Child$(v)$ for any $v\in L_2$ was $2r_2$ before the contraction procedure. Thus, we use $k_1$ balls of radius $4\cdot 2r_1+2r_2\leq 10r_1$ and $k_2$ balls of radius $0+2r_2=2r_2$ to cover at least $m$ points. This results in a 10-approximate solution for $\mathcal{I}$. Otherwise, the algorithm of Lemma~\ref{algoforcontracted} outputs that $\mathcal{I}_{\mathrm{contracted}}$ is infeasible, and by Lemma~\ref{contractionreduction} we know that $\sum_{v\in L_2}w(v)\text{cov}(v)<m$ acts as a separating inequality which is then fed back to the ellipsoid algorithm. This concludes the proof of Theorem \ref{bottomupnukc}.

Now we present the proof of Lemma~\ref{contractionreduction}. In the next subsection we present the algorithm of Lemma~\ref{algoforcontracted} and its analysis.
\begin{proof}[Proof of Lemma~\ref{contractionreduction}]
We know that $\{(\text{cov}_1(v),\text{cov}_2(v))\}_{v\in X}$ satisfies $\sum_{v\in L_2}w(v)\text{cov}(v)\geq m$. If $\{(\text{cov}_1(v),\text{cov}_2(v))\}_{v\in X}$ is in $\mathcal{P}_{\mathcal{I}}$ then there must be an integral solution $C=(C_1,C_2)$ of $\mathcal{I}$, i.e. $k_i$ balls of radius $r_i$ around points in $C_i$ for all $1\leq i \leq 2$ cover $m$ points of $X$, satisfying $\sum_{v\in L_2}w(v)\mathbbm{1}_{\{v \text{ covered by }C\}}\geq m$. We construct a solution for $\mathcal{I}_{\mathrm{contracted}}$ as follows: Move each ball of radius $r_1$ centered at some point in $C_1$ to be centered at any point in $L_2$ that it covers and increase its radius to $2r_1$. Thus, each such ball still covers all the points of $L_2$ it covered before. Let $C'_1\subseteq L_2$ be the set of centers of balls of radius $2r_1$ obtained by this procedure. Similarly, we obtain $C'_2\subseteq L_2$ by applying this procedure to $C_2$. 
However, since the pairwise distance between points of $L_2$ was strictly more than $2r_2$, each such ball of radius $r_2$ can cover at most one point of $L_2$. We can decrease its radius to 0 and it still covers all the points of $L_2$ it covered before. Let $C'=(C'_1,C'_2)$. From the construction it follows that since $C$ satisfies $\sum_{v\in L_2}w(v)\mathbbm{1}_{\{v \text{ covered by C }\}}\geq m$, $C'$ also satisfies $\sum_{v\in L_2}w(v)\mathbbm{1}_{\{v \text{ covered by } C' \}}\geq m$. Thus, $k_1$ balls of radius $2r_1$ around points in $C'_1$ and $k_2$ balls of radius 0 around points in $C'_2$ cover at least $m$ points in the instance $\mathcal{I}_{\mathrm{contracted}}$, as for any $v\in L_2$, all the $w(v)=|\text{Child}(v)|$ -many points are co-located with $v$ in $\mathcal{I}_{\mathrm{contracted}}$. This finishes the proof of the lemma.
\end{proof}
\subsection{Algorithm for contracted instances}
We now present the proof of Lemma~\ref{algoforcontracted}.
\begin{proof}[Proof of Lemma \ref{algoforcontracted}]
Recall that we are given a contracted instance of Robust 2-NU$k$C   $\mathcal{I}= ((X,d),(k_1,r_1),(k_2,r_2),m)$ and a set $L\subseteq X$ that satisfy the following properties.
\begin{enumerate}
    \item $r_2=0$.
    \item For every point $u\in X$ there is a point $v\in L$ such that $d(u,v)=0$. Furthermore, $d(v,v')>0$ for every $v,v'\in L$.
\end{enumerate}
We again use the round-or-cut method on $\mathcal{P}_{\mathcal{I}}$ using the ellipsoid algorithm. This time, when we are given fractional coverages $\{(\mathrm{cov}_1(v),\mathrm{cov}_2(v))\}_{v\in X}$ in each round we either generate a 4-approximate solution to $\mathcal{I}$, or we find an inequality separating these coverages from $\mathcal{P}_{\mathcal{I}}$ which is then fed back to the ellipsoid algorithm.

The algorithm in~\cite{DBLP:conf/ipco/ChakrabartyN21}, as well as our method, uses a standard greedy partitioning scheme $\mathcal{I}$ called CGK \cite{chakrabarty2016non}, which is just the HS procedure applied twice. Since this is a standard technique, we will not present the proofs associated with it but rather just state its guarantees in the form of a lemma that will be useful for our algorithm. CGK uses the fractional coverages $\{(\mathrm{cov}_1(v),\mathrm{cov}_2(v))\}_{v\in X}$ to obtain a tree-structured instance with properties summarized as follows.
\begin{lemma}[~\cite{chakrabarty2016non, DBLP:conf/ipco/ChakrabartyN21}]
\label{maketree}
Given an instance $\mathcal{I}=((X,d),(k_1,r_1),(k_2,r_2),m)$ of Robust 2-NU$k$C, parameters $\alpha_1,\alpha_2\geq 2$ and $\{(\mathrm{cov}_1(v),\mathrm{cov}_2(v)\}_{v\in X}$ satisfying $\sum_{v\in X}\mathrm{cov}(v)\geq m$, there is a polynomial time algorithm (CGK) that returns the following:
\begin{enumerate}
    \item sets $L_1,L_2\subseteq X$ that satisfy $d(v,v')>\alpha_i r_i$ for all $v,v' \in L_i$ for $i\in \{1,2\}$; and
    \item sets $\mathrm{Child}_1(v)\subseteq L_2 \, $ for all $ \, v\in L_1$ that partition $L_2$. Furthermore, $d(u,v)\leq \alpha_1 r_1$ for all $u\in \mathrm{Child}_1(v)$ and $v\in L_1$; and
    \item sets $\mathrm{Child}_2(v)\subseteq X \,$ for all $ \, v\in L_2$ that partition $X$. Furthermore, $d(u,v)\leq \alpha_2 r_2$ for all $u\in \mathrm{Child}_2(v)$ and $v\in L_2$.
    \item If $\sum_{v\in L_1}\mathrm{cov}_1(v)\leq k_1 - x$ for some $x\geq 0$ and $\sum_{v\in L_2}\mathrm{cov}_2(v)\leq k_2$ , then one can obtain a solution for $\mathcal{I}$ that uses $k_1+2-x$ balls of radius $(\alpha_1+\alpha_2)r_1$ and $k_2$ balls  of radius $\alpha_2r_2$. 
\end{enumerate}
Furthermore, if $\{(\mathrm{cov}_1(v),\mathrm{cov}_2(v))\}_{v\in X}$ is feasible for LP1 then $\sum_{v\in L_i}\mathrm{cov}_i(v)\leq k_i$ is satisfied for $i\in \{1,2\}$.
\end{lemma}
First, we check whether $\sum_{v\in X}\mathrm{cov}(v)\geq m$. If this is not true, then this inequality itself serves as a separating inequality. Next, we run CGK on $\mathcal{I}$ using these coverages with $\alpha_1=4$ and $\alpha_2 = 2$ to obtain  $L_1,L_2$, and $\{\mathrm{Child}_1(v)\}_{v\in L_1},\{\mathrm{Child}_2(v)\}_{v\in L_2}$.  We check if $\sum_{v\in L_i}\mathrm{cov}_i(v)\leq k_i$ for $i\in \{1,2\}$. If not, then by Lemma~\ref{maketree} these coverages are not feasible for LP1 and thus whichever inequality is violated will serve as a separating inequality. Now we branch into two cases. 
First, assume that $\sum_{v\in L_1}\mathrm{cov}_1(v)\leq k_1-2$. In this case we can apply (4) of Lemma \ref{maketree} to get a solution for $\mathcal{I}$ that opens $k_1$ balls of radius $\alpha_1 r_1+\alpha_2r_2 = 4r_1$ and $k_2$ balls of radius $\alpha_2r_2 = 0$, which is a 4-approximate solution to $\mathcal{I}$.

For the remainder of the proof we will assume that $\sum_{v\in L_1}\mathrm{cov}_1(v)> k_1-2$.
If the fractional coverages $\{(\mathrm{cov}_1(v),\mathrm{cov}_2(v))\}_{v\in X}$ are in $\mathcal{P}_{\mathcal{I}}$, then there must be an integral solution $(C_1,C_2)$ of $\mathcal{I}$ such that balls of radius $r_1$ centered around points of $S_1$ cover at least $k_1-1$ points of $L_1$.  
Moreover, since the pairwise distance between points of $L_1$ is strictly more than $\alpha_1r_1=4r_1$, each ball of radius $r_1$ of $C_1$ can cover at most one point in $L_1$. 
Thus, there is at most one ball centered at some point $v_1\in C_1$ that does not cover a point of $L_1$ and we can guess $v_1$ by enumerating over all possibilities. 
By opening balls of radius $2r_1$ around points of $L_1$ covered by balls centered at points in $C_1$, and a ball of radius $r_1$ around $v_1$, we can cover all points that are covered by balls centered at points in $C_1$. We remove the ball of radius $r_1$ around $v_1$ from the metric $X$, update $m$ by subtracting the number of points in this ball around $v_1$, and reduce $k_1$ by 1. For simplicity of notation, we still refer to these \emph{updated quantities} as $X$, $m$ and $k_1$. This now implies there is a feasible solution to the following question: Is it possible to open $k_1$ balls of radius $2r_1$ only centered at points in $L_1$, and $k_2$ balls of radius $0$ only centered at points of $L$ (this is without loss of generality as every point in $X$ is co-located with a point in $L$ by the definition of a contracted instance) to cover at least $m$ points.  We claim that this is a \emph{laminar instance} according to the following definition.
\begin{definition}\label{laminarinstance}\rm{
An instance of Robust 2-NU$k$C, $\mathcal{I} = ((X,d),(k_1,r_1),(k_2,r_2),m)$ is said to be \emph{laminar} if we are given sets $L_1,L_2\subseteq X$ such that the following are satisfied.
\begin{enumerate}
    \item The $k_i$ balls of radius $r_i$ are only allowed to be centered at points in $L_i$, $ i\in \{1,2\}$;
    \item $\mathcal{B}(u,r_i)\cap \mathcal{B}(v,r_i) = \emptyset$ for all $u,v\in L_i$, $i\in \{1,2\}$;
    \item $C(v)\cap C(v')=\emptyset$ for all $v,v' \in L_1$, where $C(v) = \{u\in L_2: \mathcal{B}(v,r_1)\cap \mathcal{B}(u,r_2)\neq \emptyset\}$ are the \emph{children} of $v$.
\end{enumerate}
}
\end{definition}
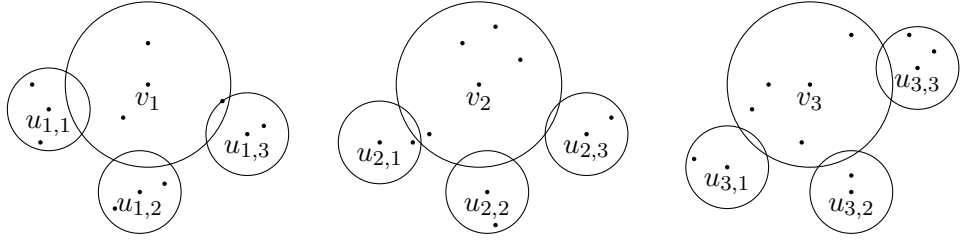
\begin{figure}\label{fig:laminarexample}
\centering
\begin{tikzpicture}[scale=1.1]
\draw (0 ,0) circle [radius=1];
\draw [fill] (0,0.5) circle [radius=0.02];
\draw [fill] (0.9,-0.2) circle [radius=0.02];
\draw [fill] (-0.3,-0.4) circle [radius=0.02];
\draw [fill] (-1.4,0) circle [radius=0.02];
\draw [fill] (-1.3,-0.7) circle [radius=0.02];
\draw [fill] (-0.4,-1.5) circle [radius=0.02];
\draw [fill] (0.2,-1.2) circle [radius=0.02];
\draw [fill] (1.4,-0.5) circle [radius=0.02];
\draw [fill] (0,0) circle [radius=0.02];
\node at (0,-0.2) {$v_1$};
\draw (4 ,0) circle [radius=1];
\draw [fill] (4,0) circle [radius=0.02];
\draw [fill] (4.2,0.7) circle [radius=0.02];
\draw [fill] (3.2, -0.7) circle [radius=0.02];
\draw [fill] (4.2, -1.7) circle [radius=0.02];
\draw [fill] (4.5, 0.3) circle [radius=0.02];
\draw [fill] (3.8,0.5) circle [radius=0.02];
\draw [fill] (3.4,-0.6) circle [radius=0.02];
\draw [fill] (5.6,-0.4) circle [radius=0.02];
\node at (4,-0.2) {$v_2$};
\draw (8 ,0) circle [radius=1];
\draw [fill] (7.5,0) circle [radius=0.02];
\draw [fill] (7.3,-0.3) circle [radius=0.02];
\draw [fill] (7.9,-0.7) circle [radius=0.02];
\draw [fill] (8.5,0.6) circle [radius=0.02];
\draw [fill] (6.6,-0.9) circle [radius=0.02];
\draw [fill] (8.5,-1.1) circle [radius=0.02];
\draw [fill] (9.5,0.4) circle [radius=0.02];
\draw [fill] (9.2,0.6) circle [radius=0.02];
\draw [fill] (8,0) circle [radius=0.02];
\node at (8,-0.2) {$v_3$};
\draw (-1.2 ,-0.3) circle [radius=0.5];
\draw [fill] (-1.2 ,-0.3) circle [radius=0.02];
\node at (-1.2,-0.5) {$u_{1,1}$};
\draw (-0.1 ,-1.3) circle [radius=0.5];
\draw [fill] (-0.1 ,-1.3) circle [radius=0.02];
\node at (-0.1 ,-1.5) {$u_{1,2}$};
\draw (1.2 ,-0.6) circle [radius=0.5];
\draw [fill] (1.2 ,-0.6) circle [radius=0.02];
\node at (1.2 ,-0.8) {$u_{1,3}$};
\draw (2.8 ,-0.7) circle [radius=0.5];
\draw [fill] (2.8 ,-0.7) circle [radius=0.02];
\node at (2.8 ,-0.9) {$u_{2,1}$};
\draw (4.1 ,-1.3) circle [radius=0.5];
\draw [fill] (4.1 ,-1.3) circle [radius=0.02];
\node at (4.1 ,-1.5) {$u_{2,2}$};
\draw (5.3 ,-0.6) circle [radius=0.5];
\draw [fill] (5.3 ,-0.6) circle [radius=0.02];
\node at (5.3 ,-0.8) {$u_{2,3}$};
\draw (7 ,-1) circle [radius=0.5];
\draw [fill] (7 ,-1) circle [radius=0.02];
\node at (7 ,-1.2) {$u_{3,1}$};
\draw (8.5 ,-1.3) circle [radius=0.5];
\draw [fill] (8.5 ,-1.3) circle [radius=0.02];
\node at (8.5 ,-1.5) {$u_{3,2}$};
\draw (9.3 ,0.2) circle [radius=0.5];
\draw [fill] (9.3 ,0.2) circle [radius=0.02];
\node at (9.3 ,0) {$u_{3,3}$};
\end{tikzpicture}
\caption{Example of a laminar instance with $L_1=\{v_1,v_2,v_3\}$ and C$(v_i)=\{u_{i,1},u_{i,2},u_{i,3}\}$ for all $i\in \{1,2,3\}$}
\end{figure}
We refer the reader to Figure \ref{fig:laminarexample} for an illustration of a laminar instance. Due to this laminar structure, we can easily solve laminar instances using standard dynamic programming techniques.
\begin{lemma}\label{laminardp}
There is a polynomial time algorithm based on dynamic programming that exactly solves a given laminar Robust 2-NU$k$C instance $\mathcal{I} = ((X,d),(k_1,r_1),(k_2,r_2),m)$  with sets of candidate centers $L_1,L_2$.
\end{lemma}
The precise details of the dynamic programming algorithm and the proof of Lemma \ref{laminardp} are given in Appendix A. Equipped with this definition and Lemma \ref{laminardp}, we now formally show why our updated contracted instance is laminar, and also how to solve it.
\begin{claim}
The \emph{updated instance} $((X,d),(k_1,2r_1),(k_2,r_2),m)$ where $r_2=0$ with sets $L_1$ and $L$ of candidate centers for balls of radius $2r_1$ and $0$ respectively, is a laminar instance.
\end{claim}
\begin{proof}
For every $v\in L_1$, let $C(v) = \{p\in L_2: \mathcal{B}(v,2r_1)\cap \mathcal{B}(p,0)\neq \emptyset\}$. We claim that for any $u,v\in L_1$, $C(v)\cap C(u)=\emptyset$. Suppose for a contradiction that there are $u,v\in \widehat{L}_1$ such that $C(v)\cap C(u)\neq\emptyset$ and let $p$ be a point in the intersection. Then $d(u,v)\leq d(u,p)+d(p,v)\leq 4r_1$ by triangle inequality, but we know that the pairwise distance of points in $L_1$ is strictly more than $4r_1$, a contradiction. This also shows that $\mathcal{B}(u,2r_1)\cap\mathcal{B}(v,2r_1)=\emptyset$ $\forall u,v \in L_1$. Also, since $d(u,v)>0$, $\mathcal{B}(u,0)\cap \mathcal{B}(v,0) = \emptyset$ for any $u,v\in L$ as per the definition of a contracted instance. We add all points $p$ in $L$ that do not appear in $C(v)$ for any $v\in L_1$ to the set $C(v)$ of an arbitrary $v\in L_1$.
\end{proof}
Thus, we can find a solution to this instance using Lemma \ref{laminardp}. A solution to this problem will either result in a 2-approximation for $\mathcal{I}$, or, if the algorithm of Lemma \ref{laminardp} returns infeasible for each guess of $v_1$, then it implies that there is no integral solution to $\mathcal{I}$ in which balls of radius $r_1$ cover at least $k_1-1$ points of $L_1$. Hence, $\sum_{v\in L_1}\mathrm{cov}_1(v)\leq k_1 -2$ separates $\{(\mathrm{cov}_1(v),\mathrm{cov}_2(v))\}_{v\in X}$ from $\mathcal{P}_{\mathcal{I}}$ which is then fed back to the ellipsoid algorithm. This finishes the proof of Lemma~\ref{algoforcontracted}.
\end{proof}
\section{Conclusion}
In this paper we developed a bottom-up framework for NU$k$C that allowed us to reduce from $t$-NU$k$C to Robust $(t-1)$-NU$k$C.
A constant approximation for 3-NU$k$C follows as a corollary from the work of CN~\cite{DBLP:conf/ipco/ChakrabartyN21}.
This bottom-up approach when applied to Robust 2-NU$k$C a gives an alternate presentation of the algorithm of ~\cite{DBLP:conf/ipco/ChakrabartyN21} with the same guarantees but has a simpler analysis. 
Thus, further progress is made towards obtaining a constant approximation for $t$-NU$k$C when $t$ is a constant. We believe that this bottom-up approach is a promising approach for proving this conjecture in full generality, which remains an exciting open problem.
\bibliographystyle{siam} 
\bibliography{refs}
\appendix
\section{Dynamic programming for laminar instances}\label{laminardpappendix}
In this section we describe the details of the dynamic programming algorithm that can be used to solve a laminar Robust 2-NU$k$C instance. Recall that a Robust 2-NU$k$C instance $\mathcal{I} = ((X,d),(k_1,r_1),(k_2,r_2),m)$ is said to be \emph{laminar} if we are given sets $L_1,L_2\subseteq X$ such that the following are satisfied.  
\begin{enumerate}
    \item The $k_i$ balls of radius $r_i$ are only allowed to be centered at points in $L_i$, $ i\in \{1,2\}$;
    \item $\mathcal{B}(u,r_i)\cap \mathcal{B}(v,r_i) = \emptyset$ for all $u,v\in L_i$, $i\in \{1,2\}$;
    \item $C(v)\cap C(v')=\emptyset$ for all $v,v' \in L_1$, where $C(v) = \{u\in L_2: \mathcal{B}(v,r_1)\cap \mathcal{B}(u,r_2)\neq \emptyset\}$.
\end{enumerate}
For any $v\in L_1$ we refer to the set $C(v)$ as the children of $v$.
Let $v_1,v_2,\ldots,v_{|L_1|}$ be an enumeration of points in $L_1$, and for each $v_i$ let $u_{1,i},u_{2,i},\ldots,u_{|C(v_i)|,i}$ be an enumeration of the points in $C(v_i)$. See Figure~\ref{fig:laminarexample} for an illustration of a laminar instance. In our dynamic programming algorithm, we will have a local table and a global table.
Since the sub-instance $\mathcal{B}(v_i,r_1)\cup (\cup_{j=1}^{|C(v_i)|}\mathcal{B}(u_{j,i},r_2))$ is completely disjoint from the rest of the instance for all $i$, we will have a local table that computes solutions to this sub instance. The global table will then be used to combine solutions to these sub instances.

The local table $\text{\sc local}$ will be indexed by $i,bit,j,k'_2,m'$. We set $\text{\sc local}[i,bit,j,k'_2,m']=1$ if and only if there is a solution that opens $k'_2$ balls of radius $r_2$ centered at points in $\{u_{1,i},\ldots,u_{j,i}\}$ and these balls cover at least $m'$ points, assuming a ball of radius $r_1$ is open at $v_i$ if and only if $bit=1$. We will compute this table bottom-up from $j=1$ to $|C(v_i)|$ for a fixed $i$ and $bit$, each time computing the table entry for all values of $k'_2, m'$ for each $j$. Now by laminarity of the instance, balls of radius $r_2$ centered at points in $C(v_i)$ are disjoint. Thus we can express $\text{\sc local}[i,bit,j,\ldots]$ in terms of $\text{\sc local}[i,bit,j-1,\ldots]$ easily by enumerating whether there is a ball of radius $r_2$ open at $u_{j,i}$ or not, and if yes then reducing the budget of points needed to be covered by the amount we gain by a ball of radius $r_2$ centered at $u_{j,i}$. We formalize this and state the recurrence with the base cases below.
\begin{align*}
    \text{\sc local}[i,bit,0,0,0] &= 1  \\
    \text{\sc local}[i,bit,0, k'_2,m'] &= 0, \quad \mbox{for any other } m',k'_2 \neq 0 \\
    \text{\sc local}[i,bit,j, k'_2,m'] &= \begin{cases}
            & \text{\sc local}[i,bit, j-1,k'_2,m'] \text{ if $u_{j,i}$ not taken,}
            \vspace{0.4cm}
            \\
            &\text{\sc local}[i,1,j-1, k'_2-1,m'-|(\mathcal{B}(u_{j,i},r_2)\setminus \mathcal{B}(v_i,r_1)| \\&\hspace{3cm}\quad \text{ if $bit=1$ and $u_{j,i}$ taken,} \vspace{0.4cm}
            \\
            &\text{\sc local}[i,0,j-1, k'_2 -1,m'-|\mathcal{B}(u_{j,i},r_2)|]\\&\hspace{3cm}\quad \text{ if $bit=0$ and $u_{j,i}$ taken.}
            \vspace{0.4cm}
        \end{cases}
\end{align*}
The recurrence for the global table $\text{\sc global}$ is described next. This table is indexed by $i,k'_1,k'_2,m'$ as follows. $\text{\sc global}[i,k'_1,k'_2,m'] = 1$ if there is a solution that opens $k'_1$ balls of radius $r_1$ centered at points in $\{v_j\}_{j=1}^{i}$ and $k'_2$ balls of radius $r_2$ centered at points in $\cup_{j=1}^{i}C(v_j)$ and covers $m'$ points.

Now by laminarity of the instance, balls of radius $r_1$ around points in $L_1$ are disjoint, $C(v)$ for all $v\in L_1$ are disjoint, and balls of radius $r_2$ around points in $\cup_{j=1}^{i}C(v_j)$ are disjoint as well. Thus we can express $\text{\sc global}[i,k'_1,k'_2,m']$ in terms of $\text{\sc global}[i-1,\ldots]$ and $\text{\sc local}[v_i,|C(v_i)|,\ldots]$ by enumerating how many balls out of $k'_1$ are open in $\{v_{j}\}_{j=1}^{i-1}$, how many in $v_i$, and similarly enumerate how the other parameters $k'_2,m'$ are divided. Formally, we set $\text{\sc global}[i,k'_1,k'_2,m']$ to 1 if any one of the following are 1, otherwise we set it to 0.
\begin{itemize}
    \item $\text{\sc local}[v_1,1,|C(v_i)|,k'_2,m' - |\mathcal{B}(v_i,r_1)|] \text{ if $i=1$ and $k'_1=1$ }$
    \item $\text{\sc local}[v_1,0,|C(v_i)|,k'_2,m'] \text{ if $i=1$ and $k'_1=0$}$
    \item $\bigvee_{k''_2, m''} \Bigg[\text{\sc global}[i-1,k'_1-1,k'_2-k''_2,m'-m''] \wedge\text{ } \text{\sc local}[v_i,1,|C(v_i)|,k''_2,m''-|\mathcal{B}(v_i,r_1)|] \Bigg]$ \\ $\text{($v_i$ taken)}$
    \item $\bigvee_{k''_2, m''} \Bigg[ \text{\sc global}[i-1,k'_1,k'_2-k''_2,m'-m'']
        \wedge\text{ } \text{\sc local}[v_i,0,|C(v_i)|,k''_2,m''] \Bigg]$\\$\text{($v_i$ not taken)}$
\end{itemize}
After we have computed the $\text{\sc local}$ table for all parameter values, we can compute the $\text{\sc global}$ table in a bottom-up fashion from $i=1$ to $|L_1|$, each time computing $\text{\sc global}[i,\ldots]$ for all values of the other parameters. If the instance is feasible then consider a feasible solution of the problem. Let $k_{1i},k_{2i}$ be the number of balls of radius $r_1$ and $r_2$ in this solution centered at points in $\{v_{j}\}_{j=1}^{i}$ and $\cup_{j=1}^{i}C(v_{j})$ respectively and let $m_{i}$ be the number of points covered by them. Then the global table value $\text{\sc global}[i,k_{1i},k_{2i},m_{i}]$ will be feasible and set to 1 for all $1\leq i\leq |L|$. Thus  $\text{\sc global}[|L_1|,k_1,k_2,m]$ will be set to 1. We can also remember the choices made while computing both the local and global tables bottom-up to also find a solution to the problem, if it exists. 
\end{document}